\documentclass[conference]{IEEEtran}
\IEEEoverridecommandlockouts

\usepackage{cite}
\usepackage{amsmath,amssymb,amsfonts,amsthm}
\usepackage{graphicx}
\usepackage{textcomp}
\usepackage{url}
\usepackage{xcolor}
\def\BibTeX{{\rm B\kern-.05em{\sc i\kern-.025em b}\kern-.08em
    T\kern-.1667em\lower.7ex\hbox{E}\kern-.125emX}}

\usepackage[lined,boxed,commentsnumbered,english]{algorithm2e}
\SetKwProg{Fn}{Procedure}{}{}
\SetKwFor{ForEach}{foreach}{do}{}
\SetKwRepeat{When}{when}{end}
\SetKwIF{If}{ElseIf}{Else}{if}{then}{else if}{else}{}
\SetKwInOut{Input}{Input}
\SetKwInOut{PARM}{Parameter}
\SetKwInOut{Output}{Output}
\SetKwRepeat{Event}{Event}{end}
\SetKwFor{Step}{Step}{}{}

 \newtheorem{theorem}{Theorem}
 \newtheorem{lemma}[theorem]{Lemma}
 \newtheorem*{theorem*}{Theorem}
 \newtheorem*{lemma*}{Lemma}
 \newtheorem{proposition}[theorem]{Proposition}
 \newtheorem*{proposition*}{Proposition}
 
 \theoremstyle{definition}
 \newtheorem{defi}[theorem]{Definition}

 \usepackage{braket}
 \usepackage{arydshln}
 \usepackage{csquotes}
 \usepackage{multirow}

\usepackage{tikz}
\usetikzlibrary{calc}
\usepackage{pgfplots}
\usepgfplotslibrary{groupplots}
\usepgfplotslibrary{dateplot}
\usepackage{filecontents}
\usepackage{caption}
\pgfplotsset{compat=1.14}

\DeclareMathOperator*{\argmin}{argmin}

\newcommand{\fast}{\mbox{``{\sf fast}''}}

\newcommand{\init}{\mbox{``{\sf init}''}}
\newcommand{\pre}{\mbox{``{\sf pre-com}''}}
\newcommand{\com}{\mbox{``{\sf com}''}}
\newcommand{\SKIP}{\mbox{{\sf SKIP}}}
\newcommand{\lockvalue}{\mbox{{\sf lockvalue}}_q}
\newcommand{\lockround}{\mbox{{\sf lockite}}_q}
\newcommand{\clock}{\mbox{{\sf clock}}_q}
\newcommand{\status}{\mbox{{\sf status}}}

\newcommand{\tmax}{t_\text{max}}
\newcommand{\RBA}{{\sf RBA}}
\newcommand{\HBA}{{\sf HBA}}
\newcommand{\negl}{{\sf negl}}

\newcommand{\Iprove}{\mathcal{I}_\text{prove}}
\newcommand{\Iveri}{\mathcal{I}_\text{veri}}
\newcommand{\Oprove}{\mathcal{O}_\text{prove}}
\newcommand{\Overi}{\mathcal{O}_\text{veri}}
\newcommand{\keygen}{\mbox{{\sf KeyGen}}}
\newcommand{\prove}{\mbox{{\sf Prove}}}
\newcommand{\veri}{\mbox{{\sf Veri}}}
\newcommand{\Hyb}{{\sf Hyb}}

\begin{document}

\title{Fair Byzantine Agreements for Blockchains}

\author{
\IEEEauthorblockN{Tzu-Wei Chao}
\IEEEauthorblockA{
Taipei, Taiwan \\
 CheshireCatNick@gmail.com
}
\and
\IEEEauthorblockN{Hao Chung}
\IEEEauthorblockA{
Taipei, Taiwan \\
 chunghaoqc@gmail.com
}
\and
\IEEEauthorblockN{Po-Chun Kuo}
\IEEEauthorblockA{
Taipei, Taiwan \\
 pk@byzantine-lab.io
}
}

\maketitle

\begin{abstract}
Byzantine general problem is the core problem of the consensus algorithm, and many protocols are proposed recently to improve the decentralization level, the performance and the security of the blockchain.
There are two challenging issues when the blockchain is operating in practice.
First, the outcomes of the consensus algorithm are usually related to the incentive model, so whether each participant's value has an equal probability of being chosen becomes essential.
However, the issues of fairness are not captured in the traditional security definition of Byzantine agreement.
Second, the blockchain should be resistant to network failures, such as cloud services shut down or malicious attack, while remains the high performance most of the time.

This paper has two main contributions.
First, we propose a novel notion called \textit{fair validity} for Byzantine agreement.
Intuitively, fair validity lower-bounds the expected numbers that honest nodes' values being decided if the protocol is executed many times.
However, we also show that any Byzantine agreement could not achieve fair validity in an asynchronous network, so we focus on synchronous protocols.
This leads to our second contribution: we propose a \emph{fair}, \emph{responsive} and \textit{partition-resilient} Byzantine agreement protocol tolerating up to 1/3 corruptions.
Fairness means that our protocol achieves fair validity. 
Responsiveness means that the termination time only depends on the actual network delay instead of depending on any pre-determined time bound. 
Partition-resilience means that the safety still holds even if the network is partitioned, and the termination will hold if the partition is resolved.

For the performance, our Byzantine agreement outdoes the state-of-art synchronous protocols. Precisely, the expected round complexity of our protocol is 6.33 rounds for the static adversary. For comparison, the protocol proposed by Abraham et al. in Financial Cryptography 2019 requires expected 10 rounds and the Algorand Agreement proposed by Chen et al. in Cryptology ePrint 2018/377 requires expected 8.2 rounds. Moreover, we conduct an experiment with 21 nodes in 10 regions spanning 3 continents on Google cloud platform, and the results show the latency of our responsive protocol is 241.79 ms. 

\end{abstract}


\section{Introduction}
\label{section:introduction}
Byzantine agreement is one of the central problems in the field of distributed algorithms and cryptography.
It also plays an important role in multiparty computation and constructing cryptocurrencies.

In 1982, Lamport, Shostak, and Pease \cite{Lamport1982} introduced the \emph{Byzantine general problem}: several generals want to make a consensus on whether they should attack or not, while some of them may be malicious. 

In this paper, we consider the following setting.
Suppose there are $n$ users, of which at most $t$ may be malicious. The malicious users may deviate from the protocol arbitrarily. Each user $q$ starts with an initial value $v_q$. All the users want to decide on one of the initial values, satisfying the following three conditions:
\begin{enumerate}
    \item \textbf{Agreement.} Two honest users never decide on different values.
    \item \textbf{Termination.} All honest users terminate in a finite time.
    \item \textbf{Validity.} The decision value must be the initial value of any node.\footnote{The original Byzantine general problem only considers the binary case. That is, the initial values can only be $0$ or $1$.
    The validity is defined as if all honest nodes start from the initial value $b \in \{0,1\}$, then all honest nodes must decide on $b$.
    Here we consider the multi-value case, and we follow the definition in \cite{Fischer1983}.
    }
\end{enumerate}
The protocol that solves such a problem is called \emph{Byzantine agreement} (BA).

\subsection{Byzantine Agreement in Blockchain}
The blockchain systems allow many mutually untrusted users to maintain a distributed ledger with consensus.
However, the long confirmation latency prevents the existing blockchain systems from many daily applications.
For example, the confirmation latency of Ethereum is about 5 to 10 minutes.
It is unrealistic to wait such a long time for micro-payment systems.

Recently, some proposals try to overcome the long latency, but it is challenging to decide who has the right to issue the blocks and to guarantee that every user shares the same ledger.
Chen and Micali \cite{Chen2016} proposed a novel blockchain system, Algorand, that solves the consensus problem by BA.
Pass and Shi \cite{PassS17} also proposed a blockchain system, Hybrid consensus, that reduces the latency by BA.
The performance and the security of such blockchain systems highly depend on the underlying Byzantine agreements,
so it is imperative to design a secure and efficient Byzantine agreement protocol under the reasonable assumptions for the Internet nowadays.

\paragraph{Fairness}
The incentive model plays an essential role in most of the blockchains. It motivates the miners and validators to execute and follow the protocol. 
It also relates to the issue and the distribution of the currency. 
Consequently, if we use BA to decide whose block (initial value) is chosen, \emph{whether each participant’s value has an equal probability of being chosen} becomes essential and directly influences the economics of the blockchain.

However, the notion of fairness is not captured in the traditional security definition (agreement, termination, validity).
Therefore, to measure the performance of BA protocols, especially in the context of blockchains, we propose a new definition of the validity, called \emph{strongly fair validity}.
Intuitively, if $n$ users join a BA, the BA protocol satisfies strongly fair validity if the probability that one's value is accepted by some honest nodes is lower-bounded by $\frac{1}{n}$ except a negligible probability.

\paragraph{Synchronous and Asynchronous Network}
An extensive literature has studied Byzantine agreement in different network models.
In a synchronous network, there is a priori known upper-bound $\lambda$ of the network delay, while an asynchronous network does not have.
For convenience, we call the BA protocols designed for the former model the \emph{synchronous BA} and the BA protocols designed for the latter model the \emph{asynchronous BA}.

When applying to the blockchain, asynchronous BAs usually outperform synchronous BAs from two aspects.
First, asynchronous BA has better resistance to network failures.
Although the network nowadays is highly reliable, network failures happen from time to time.
For example, the undersea cables break or the network services shut down for updating.
In these cases, the network delay may be much longer than the typical case and the security of a synchronous BA is not guaranteed.
Second, the performance of the synchronous protocols is limited by the upper-bound $\lambda$.
On the other hand, there is no upper-bound for the network delay in the asynchronous protocols, so the protocols proceed when enough messages are delivered, which only depends on the actual network delay.

However, the nature that the asynchronous protocols do not depend on any pre-determined time bound makes it impossible to achieve strongly fair validity\footnote{In fact, even the weakly fair validity cannot be achieved. We will elucidate it in Section \ref{section:fairness}.}.
In this paper, we show the following impossibility result.
\begin{theorem}(informal, restated in Theorem \ref{theorem:strongly})
In an asynchronous network, no Byzantine agreement tolerating some Byzantine nodes can achieve agreement, termination and strongly fair validity at the same time.
\end{theorem}


Thus, our problem is whether we can have a secure synchronous BA that achieves fair validity while enjoys the advantages of asynchronous BAs as many as possible?
The answer is positive.
In the following, we introduce two desired properties for designing synchronous BA.

\paragraph{Partition-Resilience}
Algorand agreement proposed by Chen et al. \cite{Chen2018ALGORANDAS} is a synchronous protocol.
In their work, they propose a new property, called \emph{partition-resilience}:
a Byzantine agreement protocol is \emph{partition-resilient} (PR) if the agreement always holds even if the network is asynchronous, and the termination holds if the network becomes synchronous and all the delayed messages delivered.
Notice that ``a synchronous BA with PR'' is different from ``an asynchronous BA.''
For the former, the protocol is still parameterized by a time-bound $\lambda$ and some properties\footnote{In this paper, fair validity and responsiveness in our protocols depend on $\lambda$.} other than the agreement may still rely on $\lambda$.
On the other hand, an asynchronous BA performs qualitatively the same no matter the condition of the network.

The network nowadays in highly reliable, so a synchronous BA with PR enjoys all the desired properties depending on $\lambda$ most of the time, while the agreement still holds even if the occasional failure happens.
When applying to blockchains, the agreement guarantees that the chain will not fork.
Thus, PR is a reasonable requirement of a BA protocol for building a blockchain.



\paragraph{Responsiveness}
Recently, Pass and Shi \cite{PassS17} proposed a blockchain protocol, called Hybrid consensus, whose security depends on the a priori known upper-bound $\lambda$ while the protocol proceeds as soon as the actual network delay.
In \cite{PassS17}, they defined a performance metric called \emph{responsiveness}:
a protocol is called responsive if its termination time depends only on the actual network delay $\delta$ but not on the a priori known upper-bound $\lambda$.

We borrow the same notion and apply it to Byzantine agreement.
We say a BA protocol is \emph{responsive} if all the honest nodes terminate on some values as fast as the actual network proceeds without depending on any pre-determined time bound.

\paragraph{Weakly Fair Validity}
In this work, however, we show that if a BA protocol only executes once, it is impossible to achieve both responsiveness and strongly fair validity.
Hence, we define a weaker notion of fairness, called \emph{weakly fair validity}, which captures the decided values when the BA protocol is executed many times.
When applying to blockchains, BA is usually executed once for each block.
Thus, weak fair validity is a reasonable metric if we examine the distribution of the proposers for a series of blocks.
We will formally introduce and justify the definition in Section \ref{section:fairness}.

\paragraph{Our Contributions}
To sum up, this paper has two main contributions.
First, we formalize the notion of fairness and analyze the relevant properties, including:
\begin{itemize}
    \item we define \emph{strongly fair validity}, which states that every honest node's value has a reasonable probability of being decided if the protocol is only executed once;
    \item we define \emph{weakly fair validity}, which lower-bounds the expected numbers that honest nodes' values being decided if the protocol is executed many times;
    \item we show that no BA protocol can achieve agreement, termination and weakly fair validity at the same time in an asynchronous network;
    \item we show that no BA protocol can achieve both responsiveness and strongly fair validity even in synchronous network.
\end{itemize}

Second, we propose two partition-resilient BA protocols tolerating up to $1/3$ corruptions that achieve a different level of fairness. 
The first protocol, called \RBA, achieves strongly fair validity, while the second protocol, called \HBA, achieves both responsiveness and weakly fair validity.
The two protocols not only justify the definition of fair validity but are also pragmatic and friendly for real-world implementation.
If there is no partition, \HBA ~terminates in $(4\tmax+6)\lambda$ in the worst case, $8\lambda$ in the average case and $4\delta$ in the best case, where $\tmax$ is the number of malicious users and $\delta$ is the actual network latency. 
In addition, only the pre-determined proposer needs to propose the value, so the bandwidth complexity is low. 
Even if the pre-determined proposer crashes, other users still can reach an agreement by the followed \RBA. 
In this aspect, our protocol avoids the single point of failure and resists to the DDoS attack.

Let $n$ be the number of nodes joining the protocol and $t$ be the number of malicious nodes.
Our work can be formally summarized as the following theorems.
\begin{theorem}\label{theorem:RBA}
Synchronous authenticated Byzantine agreement can achieve partition-resilience, strongly fair validity and optimal resilience $t<n/3$ with
\begin{itemize}
    \item in the best case, 5 rounds termination and $\mathcal{O}(n^2)$ communication,
    \item expected 8 rounds termination and $\mathcal{O}(n^2)$ communication,
    \item in the worst case, $4 \tmax+6$ rounds termination and $\mathcal{O}(n^2t)$ communication
    against an adaptive adversary.
\end{itemize}
\end{theorem}

\begin{theorem}\label{theorem:HBA}
Synchronous authenticated Byzantine agreement can achieve responsiveness, partition-resilience, weakly fair validity and optimal resilience $t<n/3$ with
\begin{itemize}
    \item in the best case, less than $4\delta$ termination and $\mathcal{O}(n^2)$ communication,
    \item expected less than $6.33$ rounds termination and $\mathcal{O}(n^2)$ communication,
    \item in the worst case, $4 \tmax+9 $ rounds termination and $\mathcal{O}(n^2t)$ communication
    against a adaptive adversary.
\end{itemize}
\end{theorem}

The proof of Theorem \ref{theorem:RBA} and Theorem \ref{theorem:HBA} are given in Section \ref{section:RBA} and Section \ref{section:fastba}, respectively.

\subsection{Related Work}
To the best of our knowledge, only Abraham et al. \cite{Abraham2018} discussed the fairness in the context of BA.
In that paper, they defined the \emph{quality} of a BA: the probability of choosing a value that was proposed by an honest node is at least $\frac{1}{2}$ except with negligible probability.

Their definition \cite{Abraham2018} is not sufficient when the BA is applied to blockchains.
The quality views all the honest nodes as a whole.
There may be an honest node whose value is never accepted by other nodes, which is undesired in blockchains.
On the contrary, both the strong and the weak fair validity in this paper characterize the behavior of \emph{each} honest node.


Algorand agreement \cite{Chen2018ALGORANDAS} inspires us to design a synchronous BA resisting to the network failure.
In \cite{Chen2018ALGORANDAS}, they proposed a partition-resilient BA with leader election based on verifiable random functions.
The main contribution of our protocol is that \HBA ~further achieves responsiveness while remaining partition-resilience.
Besides, a leader is elected for each iteration in Algorand's design. On the contrary, our leader election procedure is independent of the iteration index, so the nodes are not required to propose their values at each iteration.
As a result, Algorand’s BA only achieves probabilistic finality, while \RBA ~and \HBA ~both terminate in $f$ iterations in the worst case, where $f$ is the number of malicious nodes.

Therefore, without sacrificing security, \HBA ~outdoes in the aspect of performance.
In the best case, \HBA ~terminates as fast as the actual network latency; 
in the worst case, \HBA ~achieves deterministic finality.


Another important related work is practical Byzantine fault tolerance (PBFT) by Castro and Liskov \cite{pbft}.
The notion of responsiveness is emerging in their work \cite{pbft}, but it is formally defined in \cite{PassS17}.
To achieve responsiveness, there is a specific node, the primary, that can be predicted for each view.
We adopt the same method in \HBA ~for the responsiveness.

When the primary does not follow the protocol, PBFT relies on \emph{view change} to switch to the next pre-determined primary.
However, the predictable primaries are easy to be attacked, like DDoS.
The attacker may compromise a series of primaries so that the protocol may halt for a long time.
On the contrary, in \HBA, when the primary\footnote{The pre-determined node in \HBA ~is called the pioneer. See Section \ref{section:fastba}.} is malicious and does not broadcast the valid messages, 
the honest nodes will initiate \RBA, whose leader is selected by a verifiable random function.
In this case, the attacker cannot predict who will be the leader, so the protocol terminates in the constant time in expectation.
Precisely, when \RBA ~is initiated, all the honest terminate on some values in $8 \lambda$.


Hybrid consensus \cite{PassS17} proposed by Pass and Shi is a responsive blockchain protocol, where the responsiveness relies on the underlying Byzantine fault tolerance (BFT) protocol.
Briefly speaking, the participants of the underlying BFT is selected by the permissionless Nakamoto consensus since the consistency of blockchain guarantees that every honest party agrees on the same set of participants.
Hence, \HBA ~can also be adopted as the underlying BFT.

\subsection{Technical Overview of \RBA ~and \HBA}
In this paper, we propose two BA protocols. Both of them achieve partition-resilience and tolerate up to $1/3$ corruptions.
The first protocol, robust Byzantine agreement (\RBA), achieves the strongly fair validity. 
The second protocol, hybrid Byzantine agreement (\HBA), achieves responsiveness and the weakly fair validity.
In the following paragraphs, we highlight the insights on how these protocols achieve these properties.
For convenience, we set the threshold of a supermajority to be $2t+1$ out of total population $3t+1$, where $t$ is the number of Byzantine nodes.

\paragraph{Agreement}
\RBA~is a leader based protocol.
The leader is elected by the pseudorandom value of a verifiable random function, which we called a credential.
At the beginning of \RBA, each node proposes its value and the credential for being a leader.

Then, each iteration consists of two phases of voting:
In the first phase voting, nodes identify the leader by comparing their credential and vote on the leader's value.
If an honest node receives a supermajority of votes for the same value, 
the node locks the value.
In the second phase voting\footnote{In Section \ref{section:RBA}, we call the first phase voting \emph{pre-commit message} and call the second phase voting \emph{commit message}.}, the nodes vote for the locked values.

If a node locks on some value, it will always vote for the locked value in the following iterations,
unless the locked value is updated.
A node only locks one value and updates its locked value if it receives a supermajority of votes for the same value in the first phase in the future iteration.

A node terminates if a supermajority of votes for the same value in the second phase.
This means that there is a supermajority of nodes locks on the value.

\paragraph{Partition-Resilience}
We design two mechanisms to achieve this.
First, at any time, at most one value can be locked by a supermajority of nodes.
Once the supermajority of node locks on a value, all honest nodes in the supermajority will only vote on the value for the first phase in the following iterations.
Then, it is impossible that a new value will be locked.
Hence, the honest nodes never decide on different values even if the partition exists.

Second, to ensure node can process in the same iteration even network partition sometimes happened, nodes will jump to the newest iteration if it receives a majority of votes in the first phase of that iteration.
That is, \textit{each node will update the locking value not only by the timing bound from the synchronous network but also the condition of valid votes is received asynchronously to against network partition}.

\paragraph{Responsiveness}
For \HBA, the pre-determined leader (we called \emph{pioneer}) mechanism is adopted.
Each node can know who is the pioneer by some pre-determined information before the protocol starts.
In the first iteration of \HBA, each node votes the value proposed by the pioneer \emph{immediately}.
If the pioneer is honest and the network operates normally, all the nodes simply decide on pioneer's value.
Otherwise, if no value is decided after the first iteration, every node starts \RBA ~with the initial state inherit from the first iteration.

Since the pioneer is pre-determined in the first iteration, each node decides on pioneer's value as soon as the votes in the first and second phase are enough.
In other words, the nodes work asynchronously in the first iteration, and thus,
the latency only depends on the actual network instead of the upper-bound.
Note that, there are still two voting phases in case of a network partition.

\paragraph{Fair Validity}
In \RBA, every node follows the leader's value, so we have to make sure every node has a reasonable probability of being chosen as the leader for the strongly fair validity.
The leader is chosen by the credentials from each node.
Thus, nodes have to wait for the worst-case network latency to ensure all the messages from honest nodes are received.

On the other hand, in \HBA, every node follows the pioneer's value to achieve the responsiveness,
so other node's value will not be decided if the pioneer and the network work normally.
Thus, \HBA ~only achieve weakly fair validity.
To do this, the pioneer election is done by permutation.
That is, there is a deterministic list for the order of pioneers (e.g., ranking by public key).
Suppose there are $n$ nodes joining the protocol.
In this case, the expected numbers that honest nodes' values being decided are roughly $\frac{M}{n}$ after $M$ times of the protocol.

\paragraph{Optimal Resilience}
The famous results by Dwork et al. \cite{DworkLS84} showed the impossibility of a permissioned consensus protocol even with public key infrastructure cannot tolerate 1/3 or more Byzantine corruptions in an asynchronous network.
Conceptually, suppose $n$ nodes are divided into three distinct sets: $S_1$, $S_2$ and $S_B$, where the nodes in $S_1$ and $S_2$ are honest and the ndoes in $S_B$ are malicious.
Due to the asynchronous network, the messages between $S_1$ and $S_2$ are delayed arbitrarily long.
Without loss of generality, we assume $n = 3t+1$ and $|S_1| = |S_2| = t, |S_B| = t+1$.
If the protocol only needs $2t$ nodes to proceed, the nodes in $S_B$ can send inconsistent messages to $S_1$ and $S_2$.
Then, the nodes in $S_1$ and $S_2$ may agree on the different values, respectively, so the agreement breaks.
If the protocol needs more than $2t$ nodes to proceeds, the nodes in $S_B$ just crash.
Then, the protocol will halt forever and the termination breaks.

Thus, once the malicious users are more than $\lfloor \frac{n-1}{3}\rfloor$, either the agreement or the termination breaks.
To achieve partition-resilience, both \RBA ~and \HBA ~tolerates $1/3$ corruptions, which is the optimal according to the argument above.\\



%
\subsection{Roadmap}
In Section \ref{section:model}, we formalize our network and adversary models. 
In Section \ref{section:fairness}, we define strongly fair validity and weakly fair validity, and we also give two relevant impossibility results.
Then, the protocol and the security analysis of \RBA ~and \HBA ~are given in Section \ref{section:RBA} and Section \ref{section:fastba}, respectively.
We analyze the communication complexity of \RBA ~and \HBA ~in Section \ref{section:complexity}.
We implemented our protocols by Go language and deployed on Google Cloud Platform services. The experiment results are presented in Section \ref{section:experiment}.
We also compare the performance of \HBA ~and three other BA protocols under different network conditions by simulation in Section \ref{section:simulation}.
Finally, the main contributions are concluded in Section \ref{section:conclusion}.
The proofs in Section \ref{section:RBA} and Section \ref{section:fastba} are given in Appendix \ref{app:rba} and Appendix \ref{app:hba}, respectively.

\section{Preliminaries}
\label{section:model}

\subsection{System Model}

In this paper, we consider the authenticated setting (i.e., digital signature and public key infrastructure (PKI) exist).
We further assume that when the users register their public keys on the PKI, they cannot choose the key in favor of other users' keys.
In practice, this can be done by commit-and-reveal schemes.
The users register the hash values of their public keys on the PKI first.
After all the users have registered, they reveal the public keys.

We say the adversary is \emph{static} if the adversary has to choose which nodes are corrupted before the protocol starts.
On the contrary, we say the adversary is \emph{adaptive} if the adversary can choose which nodes are corrupted during the protocol.
The corrupted nodes are called \emph{Byzantine} and the nodes that are not corrupted are called \emph{honest}. 
A Byzantine node can deviate from the protocol arbitrarily; it can engage in problematic malfunctions such as sending conflicting messages, violating algorithm criteria, delaying the messages between other nodes, and so on. We also assume the adversary has full control of the network. 
The adversary can learn all the messages delivered on the network and determine the delay and the order of the delivered messages.

We say a network is \emph{synchronous} if there exists a known time bound $\lambda$.
We say the network is \emph{partitioned} if the messages between the honest nodes are delayed such that the delivering time exceeds $\lambda$. 
A network is \emph{recovered from partition} if all nodes receive all the previous messages which should be delivered and the delay of the message is lower than $\lambda$.
We say a network is \emph{asynchronous} if time bound doesn't exist.

\subsection{Terminology}
Let $X$ be a random variable. The expectation value of $X$ is denoted as $\mathbb{E}[X]$.

A function $f$ is \emph{negligible} if for all polynomial $p$, there exists an integer $N$ such that for all integers $n > N$, it holds that $f(n) < \frac{1}{p(n)}$.

Given a security parameter $\kappa$, a protocol $P_0$ and a protocol $P_1$, we say $P_0$ is \emph{indistinguishable} from $P_1$, if for all polynomial-time distinguishers $D$, there is a negligible function $f$ such that \[
\underset{b\leftarrow \{0,1\}}{\Pr}[b=b'|D \text{ operates in } P_b(\kappa) \text{ and } D \text{ outputs } b'] \leq f(\kappa).
\]

\subsection{Verifiable Random Function}
The verifiable random function (VRF), introduced by Micali, Rabin and Vadhan \cite{vrf}, is a type of pseudorandom function by which anyone can verify the validity of the function evaluation from public information. 
The formal definition is given in Appendix \ref{appvrf}.

\section{Fairness}
\label{section:fairness}
\subsection{Definition}
Let $\kappa$ be the security parameter and $n$ be the number of total nodes joining Byzantine agreements, where each node $q_i$ starts with the initial value $v_i$.
Let $\mathcal{H}$ be the set of honest nodes at the end of the Byzantine agreement.

\begin{defi}[strongly fair validity]
A Byzantine agreement achieves \emph{strongly fair validity} for a set of adversaries $\mathcal{A}$ if for all adversaries in $\mathcal{A}$, there exists a negligible function $\negl$ such that for all $q_i \in \mathcal{H}$, it holds that
\begin{align}
    \Pr[v_i \text{ is the decided by some honest node}] \geq \frac{1}{n}-\negl(\kappa).
\end{align}
\end{defi}

In practice, if we apply the Byzantine agreement to blockchains, the Byzantine agreement may be executed many times, once for each block. 
Hence, we propose a weaker version of fairness, called \emph{weakly fair validity}.
Intuitively, it guarantees the lower-bound of the expected numbers that honest nodes’ values being decided if the protocol is executed many times.

\begin{defi}[weakly fair validity]
Suppose we execute a Byzantine agreement $M$ times.
Let $X_{ij}$ be a binary random variable such that $X_{ij} = 1$ if $q_i$'s initial value is decided by some honest node in $j$-th time; otherwise, $X_{ij} = 0$.
Let $\mathcal{H}_M$ be the set of honest nodes when the Byzantine agreement has to be executed $M$ times.
Then, we say the Byzantine agreement achieves \emph{weakly fair validity} for a set of adversaries $\mathcal{A}$ if, for all adversaries in $\mathcal{A}$, there exists a negligible function $\negl$ such that for all $q_i \in \mathcal{H}_M$, it holds that
\begin{align}
    \mathbb{E}[X_i] \geq \lfloor\frac{M}{n}\rfloor-\negl(\kappa),
\end{align}
where $X_i = \sum_{j=1}^M X_{ij}$.
\end{defi}

Obviously, a BA with strongly fair validity must be a BA with weakly fair validity.

\subsection{Impossibility of fair Byzantine agreements}
In this section, we prove two impossibilities of fair Byzantine agreements.

\begin{theorem}\label{theorem:asynValidity}
In an asynchronous network, no Byzantine agreement tolerating some Byzantine nodes can achieve agreement, termination and weakly fair validity at the same time.
\end{theorem}
\begin{proof}
Let $n$ be the number of total nodes joining Byzantine agreements.
We divide $n$ nodes into two sets: $n-1$ nodes in the first set $S_1$ and one node in the second set $S_2$.
Due to the asynchronous network, the delay between two sets can be arbitrarily long while the messages delivered in the same set arrive immediately.
In this case, the nodes in $S_1$ cannot distinguish whether the node in $S_2$ is Byzantine or the network is partitioned.
If the nodes in $S_1$ wait for the messages from $S_2$, the termination fails because the node in $S_2$ may be Byzantine.
If the nodes in $S_1$ do not wait for the messages from $S_2$, the weakly fair validity fails because the initial value of the node in $S_2$ will not be considered in the protocol.
\end{proof}

Theorem \ref{theorem:asynValidity} rules out the possibility that constructing an asynchronous Byzantine agreement to achieve fairness. 
That is why we construct \RBA~in section \ref{section:RBA} to achieve fairness.

Since the latency of typical synchronous Byzantine agreements is a multiple of the worst-case bound of network latency.
The latency of fair Byzantine agreements is also bounded by the worst-case network latency.
Can we construct a responsive Byzantine agreement that achieves
fairness?
We prove the impossibility of this question in next theorem.
\begin{theorem}\label{theorem:strongly}
Responsive synchronous Byzantine agreements cannot achieve strongly fair validity.
\end{theorem}
\begin{proof}
We prove this theorem by contradiction. 
Assume a responsive synchronous Byzantine agreement achieves strongly fair validity.
Let $u$ be an honest node in the responsive synchronous Byzantine agreement and the latency of message sent from $u$ to other nodes be always in $(\lambda - \epsilon, \lambda)$, where $\epsilon$ is an arbitrary positive number less than $\lambda$.
If the decided time be in $(0,\lambda - \epsilon)$,
the message of the honest node $u$ has zero probability of being decided by the responsive Byzantine agreement.
This violates the definition of strongly fair validity.
Otherwise, the decided time is larger than $\lambda - \epsilon$.
This violates the definition of responsiveness.
\end{proof}

Thus, we construct a responsive Byzantine agreement \HBA~to achieve weakly fair in section \ref{section:fastba} as an example that a responsive Byzantine agreement can achieve weakly fairness.

\section{Robust Byzantine Agreement}
\label{section:RBA}
In this section, we introduce our first Byzantine agreement protocol with partition-resilience, strongly fair validity tolerating up to $1/3$ corruptions.
We call the protocol \emph{robust Byzantine agreement} ({\sf RBA}).


\subsection{Protocol}
\label{subsection:baprotocol}

Let $S$ be the set of all nodes and $n$ be the size of $S$.
Let $\tmax = \lfloor (n-1)/3 \rfloor$ and $V$ denote the set of values that can be decided. We also define two special values $\bot$ and $\SKIP$ that are not in $V$. 
For each node $q \in S$, $q$ has four internal variables: $r_q$ records the index of the iteration at which $q$ is working, $\lockvalue$ records the candidate value that $q$ supports, $\lockround$ records the index of the iteration from which $\lockvalue$ comes and $\clock$ is $q$'s local clock. Let $sk_q$ and $pk_q$ denote the secret key and public key of $q$, respectively. 

Let $F$ be a verifiable random function (VRF).
We write \[
\langle y,\pi \rangle \leftarrow F_{sk}(m)
\]
to denote the output of $F$ on the message $m$ with the secret key $sk$, where $y$ is the pseudorandom value and $\pi$ is the proof for $y$.
We define $\status$ to be the pre-determined public information, for example, the public key of all nodes in BA or the height of blocks in blockchain.

\paragraph{Message types}
We define three kinds of messages:
\begin{enumerate}
    \item the \emph{initial message} of the node $q$: $\left(\init,v_q,q,\langle y_q,\pi_q \rangle \right)$, where $\langle y_q,\pi_q \rangle \leftarrow F_{sk_q}(\status)$
    \item the \emph{pre-commit message} of the value $v$ from the node $q$ at the iteration $r$: $\left(\pre,v,q,r\right)$
    \item the \emph{commit message} of the value $v$ from the node $q$ at the iteration $r$: $\left(\com,v,q,r\right)$
\end{enumerate}
We assume all messages are protected by the digital signature, so the authentication of the messages hold.

\paragraph{Leader election}
With these notations, we introduce the leader election algorithm which will be a subroutine of \RBA. Let $M_q$ denote the set of initial messages that the node $q$ receives from all the nodes (including $q$ itself). The node $q$ verifies the VRF value $\langle y_j,\pi_j \rangle$ in $M_q$ and sets $U_q$ to be the set of nodes whose VRF values are valid. Then, $q$ computes \[
\ell_q = \argmin_{j\in U_q} y_j.
\]
We say $\ell_q$ is the \emph{leader} of $q$.

\paragraph{Updating internal variables}
Suppose a node $q$ is working at iteration $r_q$. The node $q$ updates its internal variables \emph{as soon as} one of the following conditions holds:
\begin{enumerate}
    \item (lock condition) If node $q$ has seen $2\tmax +1$ pre-commit messages of the same value $v \in V \cup \{\bot\}$ at the same iteration $r$ such that $r = r_q$, $q$ sets $\lockvalue = v$ and $\lockround = r$.
    \item (forward condition) If node $q$ has seen $2\tmax +1$ pre-commit messages of the same value $v \in V \cup \{\bot\}$ at the same iteration $r$ such that $r > r_q$, $q$ sets $\clock = 2\lambda$ and starts the iteration $r$ from Step 2.
    \item (forward condition) If the node $q$ has seen $2\tmax +1$ commit messages of any value at the same iteration $r$ such that $r \geq r_q$, $q$ sets $\clock = 2\lambda$ and starts the iteration $r+1$ from Step 2. 
\end{enumerate}
We say that node $q$ achieves the \emph{lock condition}, if the condition 1 holds.
We say that node $q$ achieves the \emph{forward condition} if the condition 2 or the condition 3 holds. Node $q$ goes into the next iteration immediately if it achieves the forward condition even if it does not achieve the forward condition at Step 4.

\paragraph{Protocol description}
\RBA~(Algorithm \ref{algorithm:ba}) is an iteration-based protocol. 
Initially, for all honest nodes $q \in S$, $q$ initializes its internal variables by $r_q = 1$, $\lockvalue = \SKIP$, $\lockround = -1$ and $\clock = 0$ and also chooses its initial value $v_q \in V$. 

Our protocol has four steps in each iteration. At Step 1, all the nodes broadcast their own initial value $v_q$ in the format $(\init,v_q,q,\langle y_q,\pi_q \rangle)$. 

When $\clock = 2\lambda$, $q$ enters Step 2. If $\lockvalue \in \{\SKIP,\bot\}$, $q$ verifies the initial messages it receives and computes the set $U_q$ of nodes whose VRF values are valid. If $U_q \neq \emptyset$, $q$ identifies its leader $\ell_q$ and pre-commits $\ell_q$'s value; otherwise, $q$ pre-commits $\bot$. If $\lockvalue \not \in \{\SKIP,\bot\}$, node $q$ pre-commits $\lockvalue$. We say node $q$ pre-commits on a value $v$ if node $q$ broadcasts the message $(\pre,v,q,r)$ where $r$ is the iteration index that node $q$ is working at.
Note that node $q$ updates its $\lockvalue$ and $\lockround$ immediately if the lock condition holds.

When $\clock = 4\lambda$, node $q$ enters Step 3. 
Node $q$ commits on its current $\lockvalue$.
We say node $q$ commits on a value $v$ if node $q$ broadcasts the message $(\com,v,q,r)$ where $r$ is the iteration index that $q$ is working at. After node $q$ broadcasts the commit message, $q$ enters Step 4, at which $q$ waits for the forward conditions.

\paragraph{Termination condition}
Node $q$ \emph{decides} on a value $v$ \emph{as soon as} node $q$ has seen $2\tmax +1$ commit messages of the same value $v \in V \cup \{\bot\}$ at the same iteration $r$.

The protocol for a node $q$ is summarized as Algorithm \ref{algorithm:ba}.
\begin{algorithm}[!htbp]
\caption{Robust Byzantine Agreement for node $q$}
\label{algorithm:ba}
    \Input{an initial value $v_q \in V$ from node $q$ and the public key $\{pk_q\}$ from all nodes}
    \Output{an agreed value $v_{fin}\in V \cup \{\bot\}$ from some node}
  Initialize $r_q = 1$, $\lockvalue = \SKIP$, $\lockround = -1$ and $\clock = 0$
  
  \Step { 1: when $\clock = 0$,}{
    broadcast$(\init,v_q,q,\langle y_q,\pi_q \rangle)$
  }
  \Step { 2: when $\clock = 2\lambda$, }{
    \If{$\lockvalue \in \{\SKIP,\bot\}$ and $U_q \neq \emptyset$}{
      node $q$ identifies its leader $\ell_q$ at $q$'s current view\\
      broadcast$(\pre,v_{\ell_q},q,r_q)$
    }
    \ElseIf{$\lockvalue \in \{\SKIP,\bot\}$ and $U_q = \emptyset$}{
      broadcast$(\pre,\bot,q,r_q)$
    }
    \Else{broadcast$(\pre,\lockvalue,q,r_q)$}
  }
  \Step { 3: when $\clock = 4\lambda$,}{
      broadcast$(\com,\lockvalue,q,r_q)$
  }
  \Step { 4: when $\clock  \in (4\lambda,\infty)$}{
    wait until the \emph{forward condition} is achieved
  }
\end{algorithm}

\subsection{Agreement}
\label{subsection:agreement}
We first show that our protocol will reach agreement; that is, two honest nodes never decide on the different values.

\begin{lemma}\label{lemma:agreement1}
Assume $t \leq \tmax$. Suppose a node $p$ receives $2\tmax +1$ commit messages of $v_p$ and another node $q$ receives $2\tmax +1$ commit messages of $v_q$. If both these $2\tmax +1$ commit messages all come from the iteration $r$, then $v_p = v_q$.
\end{lemma}

\begin{theorem}[Agreement]
\label{thm:baagreement}
Assume $t \leq \tmax$ and the adversary is adaptive. Regardless of partition, if an honest node $p$ decides on some value $v_p$ and another honest node $q$ decides on some value $v_q$, then $v_p = v_q$. That is, the honest nodes will never decide on different values.
\end{theorem}

\subsection{Termination}\label{subsection:baTermination}
We now analyze when the algorithm terminates if no partition exists or if the system recovers from a previous partition in different adversary model.

\begin{proposition}[Termination without partition in adaptive adversary]\label{proposition:terminationWOpartition}
Assume $t \leq \tmax$ and the adversary is adaptive. If all the honest nodes start at the $r$-th iteration within time $\lambda$ and no partition exists, all the honest nodes will decide on some values in $t+1$ iterations. 
\end{proposition}

Note that the honest nodes will broadcast commit messages when $\clock = 4\lambda$. If all the honest nodes start \RBA~simultaneously, they will all receive the commit messages when $\clock = 5\lambda$. However, we allow that any two honest nodes start the protocol with at most $\lambda$ time difference. 
Thus, the node who start the protocol earliest will receive the commit messages when its local time $\clock = 6\lambda$.
That is, the best termination time of \RBA ~is $6\lambda$.

As we shown in Proposition \ref{proposition:terminationWOpartition}, all the honest nodes will terminate in $t+1$ iterations with certainty.
According to the forward conditions, nodes start the $r$-th iteration from $\clock = 2\lambda$ for all $r \geq 2$.
Hence, the first iteration costs $6\lambda$ to complete, but the following iterations only cost $4\lambda$ for each.
Thus, the $t+1$ iterations cost $(4t + 6)\lambda$ in the worst case.

\begin{proposition}[Expected termination in static adversary]
\label{proposition:terminationExpected}
Assume $t \leq \tmax$ and the adversary is static.
Suppose all honest nodes start at $r$-th iteration within time $\lambda$ and no partition exists. Then, it is expected that all honest nodes will decide on some values in $8$ rounds.
\end{proposition}

\begin{proposition}[Fast recovery from partition in adaptive adversary]
\label{proposition:partitionResolved}
Assume $t \leq \tmax$ and the adversary is adaptive.
If the partition is resolved, all the honest nodes will decide on some values in $t+2$ iteration. 
If the adversary is static, it is to be expected that all honest nodes will decide on some values in $12$ rounds.
\end{proposition}

\subsection{Strongly Fair Validity}
\label{subsection:strongfair}
In this section, we prove that \RBA ~achieves strongly fair validity.
Intuitively, when the network operates synchronously, every honest nodes receives all the initial messages from each other before $2\lambda$.
Then, as long as the underlying VRF is secure, the probability of being the leader is approximate the uniform distribution, so \RBA ~achieves strongly fair validity.
The result is formalized as the following theorem.

\begin{theorem}[strongly fair validity]
\label{thm:rbafair}
Suppose the network is synchronous and $F$ is a secure VRF.
Then, \RBA ~achieves strongly fair validity under the assumption of static adversary.
\end{theorem}

\section{Hybrid Byzantine Agreement}
\label{section:fastba}
In the previous section, the leader is selected by the lowest value of VRF in \RBA.
This limits the performance of \RBA~since nodes have to wait for the worst-case network latency to ensure all the messages from honest nodes are received.
In this section,
we improve the efficiency by hybridizing of \RBA~and the pre-determined leader method. 
Intuitively, our protocol consists of two phases.
At the beginning of the protocol, all the nodes can uniquely identify a particular node, called \emph{pioneer}, in a deterministic way.
In the first phase, the pioneer broadcast its value and other nodes pre-commit on pioneer's value as soon as possible.
If the pioneer does not propose a value, other nodes start \RBA, if timeouts.

Due to the hybrid structure, we call the improved BA in this section \emph{hybrid Byzantine agreement} (\HBA).
Note that, since there is only one pioneer, each node can decide when enough votes are received instead of waiting for the worst-case network latency.
Thus, this technique achieves \textit{responsiveness}.

\subsection{Protocol}
Now we formally introduce \HBA($p$) protocol, where $p$ is the parameter related to the \textit{pioneer election} in following part.

\paragraph{Message types}
Except for initial message, pre-commit message and commit message, we need the fourth type of message in \HBA:
\begin{enumerate}
    \item[4.] the \emph{fast message} of the pioneer node $q$: $\left(\fast,v_q\right)$
\end{enumerate}

\paragraph{Pioneer election}
Let $u_i$ be the $i^{th}$ node by sorting all the user according to their public key.
The pioneer is node $u_p$, where $p$ is the parameter of \HBA~protocol.


\paragraph{Leader election}
The leader election of \HBA~is the same as \RBA.

\paragraph{Updating internal variables}
The conditions for updating internal variables in \HBA~are almost the same in \RBA, except that the forward conditions start the next iteration from Step 4 and $\clock = 5\lambda$.
Suppose a node $q$ is working at iteration $r_q$. The node $q$ updates its internal variables \emph{as soon as} one of the following conditions holds:
\begin{enumerate}
    \item (lock condition) If node $q$ has seen $2\tmax +1$ pre-commit messages of the same value $v \in V \cup \{\bot\}$ at the same iteration $r$ such that $r = r_q$, $q$ sets $\lockvalue = v$ and $\lockround = r$.
    \item (forward condition) If node $q$ has seen $2\tmax +1$ pre-commit messages of the same value $v \in V \cup \{\bot\}$ at the same iteration $r$ such that $r > r_q$, $q$ sets $\clock = 5\lambda$ and starts the iteration $r$ from Step 4.
    \item (forward condition) If the node $q$ has seen $2\tmax +1$ commit messages of any value at the same iteration $r$ such that $r \geq r_q$, $q$ sets $\clock = 5\lambda$ and starts the iteration $r+1$ from Step 4. 
\end{enumerate}

\paragraph{Protocol description}
Before Step 1, the pioneer can be uniquely determined by all the nodes according to the pioneer election.

At Step 1, the pioneer broadcast its own initial value in the format $(\fast,v_q)$.
For every non-pioneer node $q$, $q$ starts \HBA~from Step 2.
At Step 2, the pioneer broadcasts the pre-commit message of its value.
For every node $q$, when $q$ receives pioneer's fast message, $q$ broadcasts the pre-commit message $(\pre,v_{leader},q,\infty)$ immediately if $\clock \leq 3\lambda$.

For every node $q$, if $q$ receives $2\tmax +1$ pre-commit messages of the same value $v\in V$ and $\clock \leq 3\lambda$, $q$ updates the variables $\lockvalue = v$ and $\lockround = 0$ (according to the lock condition) and broadcasts the commit message $broadcast(\com,\lockvalue,q,\infty)$ immediately.
Note that the honest nodes broadcast the pre-commit and commit messages as soon as the conditions are satisfied. 
It is the core idea to achieve the responsiveness.

Step 3 to Step 6 is a \RBA~protocol, except that the internal variables $\lockvalue$ and $\lockround$ may be changed in Step 2.

\paragraph{Termination condition}
The termination condition in \HBA~is the same as \RBA:
Node $q$ \emph{decides} on a value $v$ \emph{as soon as} node $q$ has seen $2\tmax +1$ commit messages of the same value $v \in V \cup \{\bot\}$ at the same iteration $r$.

The protocol for a node $q$ is summarized as Algorithm \ref{algorithm:hba}.

\begin{algorithm}[!htbp]
\caption{Hybrid Byzantine Agreement for node $q$}
\label{algorithm:hba}
    \PARM{a pioneer number $p$}
    \Input{an initial value $v_q \in V$ from node $q$ who has own secret key $sk_q$  and the public key $\{pk_i\}_i$ from all nodes, and public randomness and information}
    \Output{an agreed value $v_{fin}\in V \cup \{\bot\}$ from some node}
  Initialize $r_q = 0$, $\lockvalue = \SKIP$, $\lockround = -1$ and $\clock = 0$\\
  Identify the pioneer\\

  \Step { 1: when $\clock = 0$,}{
    \If{node $q$ is the pioneer}{
    $broadcast(\fast,v_q)$\\
  }
  }
  
  \Step { 2: when $\clock \leq 3\lambda$, }{
  \If{receiving pioneer's value $v_{pioneer}$ \text{or} node $q$ is the pioneer, }{
      check then $broadcast(\pre,v_{pioneer},q,0)$}
  
    \If{node $q$ sees $2\tmax +1$ pre-commit messages of the same value $v\in V$}{
      $broadcast(\com,\lockvalue,q,0)$\\
      }
  }

  \Step { 3: when $\clock = 3\lambda$,}{
    broadcast$(\init,v_q,q,\langle y_q,\pi_q \rangle)$\\
    set $r_q = 1$
  }
  \Step { 4: when $\clock = 5\lambda$, }{
    \If{$\lockvalue \in \{\SKIP,\bot\}$ and $U_q \neq \emptyset$}{
      node $q$ identifies its leader $\ell_q$ at $q$'s current view\\
      broadcast$(\pre,v_{\ell_q},q,r_q)$
    }
    \ElseIf{$\lockvalue \in \{\SKIP,\bot\}$ and $U_q = \emptyset$}{
      broadcast$(\pre,\bot,q,r_q)$
    }
    \Else{broadcast$(\pre,\lockvalue,q,r_q)$}
  }
  \Step { 5: when $\clock = 7\lambda$,}{
    broadcast$(\com,\lockvalue,q,r_q)$
  }
  \Step { 6: when $\clock  \in (7\lambda,\infty)$}{
    wait until the \emph{forward condition} is achieved
  }
\end{algorithm}


\subsection{Agreement}
\begin{theorem}[Agreement of \HBA]
\label{theorem:agreementhba}
Assume $t \leq \tmax$ and the adversary is adaptive. Regardless of partition, if an honest node $p$ decides on some value $v_p$ and another honest node $q$ decides on some value $v_q$, then $v_p = v_q$. That is, the honest nodes will never decide on different values.
\end{theorem}

\subsection{Termination}

\begin{proposition}[Termination without partition in static adversary]
\label{proposition:staticTermination}
Assume $t \leq \tmax$ and the adversary is static. 
If all the honest nodes start \HBA~within time $\lambda$ and no partition exists, all the honest nodes will decide on some values in $4\lambda$, $6.33\lambda$ and $t+1$ iterations in the best case, the average case and the worst case, respectively.
\end{proposition}

\begin{proposition}[Termination without partition in adaptive adversary]\label{proposition:terminationWOpartitionfastBA}
Assume $t \leq \tmax$ and the adversary is adaptive. 
If all the honest nodes start \HBA~within time $\lambda$ and no partition exists, all the honest nodes will decide on some values in $t+1$ iterations.
\end{proposition}

\begin{proposition}[Fast recovery from a partition in adaptive adversary]
\label{proposition:hbafastrecoveryfrompartition}
Assume $t \leq \tmax$ and the adversary is adaptive. If the partition is resolved, all the honest nodes will decide on some values in $t+2$ iterations. 
If the adversary is static, it is to be expected that all honest nodes will decide on some values in $12$ rounds.
\end{proposition}

\subsection{Responsiveness}
The following proposition directly implies that \HBA~is responsive.
\begin{proposition}
\label{proposition:hbaresponsiveness}
Assume the actual network delay is $\delta$, and all the nodes start \HBA~within $\tau$ time difference. 
If there is no partition and the pioneer is honest, all the honest nodes will decide on some values in $\tau + 3\delta$.
\end{proposition}

\subsection{Weakly Fair Validity}

\begin{theorem}[weakly fair validity]
\label{thm:hbafairness}
Suppose the network is synchronous.
Then, \HBA ~achieves weakly fair validity under the assumption of static adversary.
\end{theorem}

\section{Experiment}
\label{section:experiment}
We implemented our protocol by Go language and deployed on Google Cloud Platform services.
We ran \RBA~ and \HBA~ on 21 GCP instances (4 vCPU and 8GB RAM) uniformly distributed throughout its 10 regions spanning 3 continents.

We set $\lambda = 0.5$ second by the reason of our experiment on the latency on GCP.
However, it is optimistic to set this bound for a general network.
For \RBA~, the experiment repeats 703 times and the histogram is shown in Figure. 1. 
The average latency of \RBA~ is 3.20 seconds and the standard deviation is 87.60 ms.
The results show the latency of \RBA~ is expected 6.4 round, which is close to the round complexity of the best case.

\begin{figure}[ht]
\label{fig:rbalatency}
\centering
\includegraphics[width=6cm]{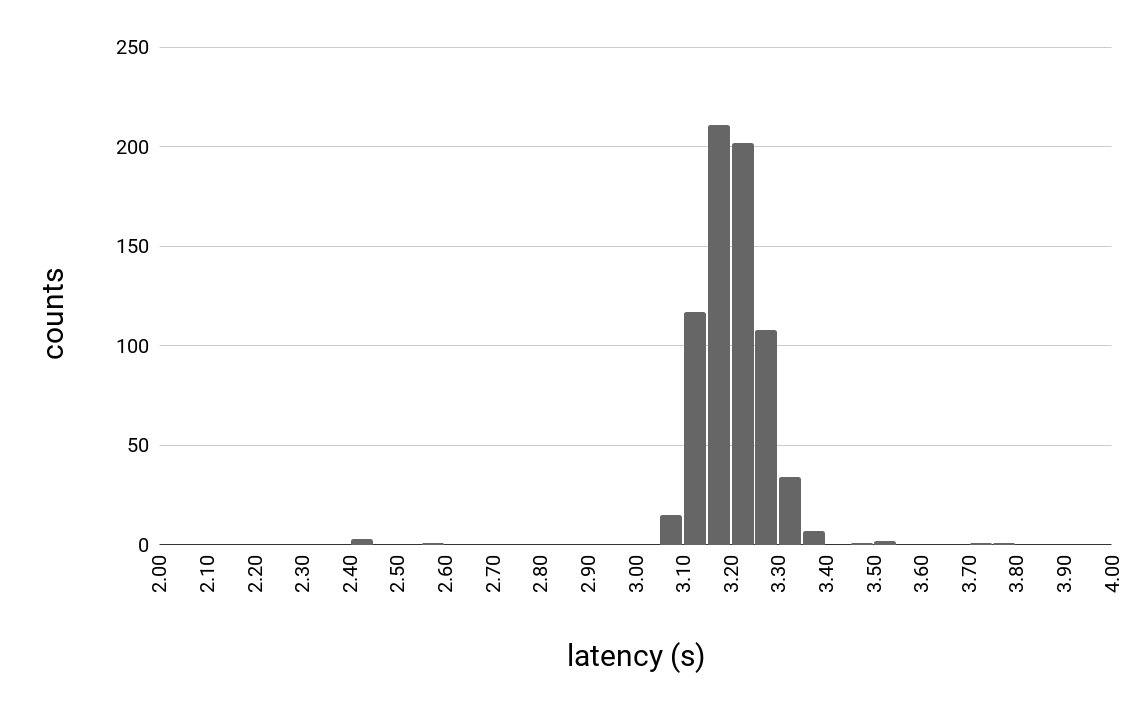}
\caption{The histogram of latency of \RBA.}
\end{figure}

For \HBA~, the experiment repeats 715 times and the histogram is shown in Figure. 2. 
This result confirms the responsiveness of \HBA~.
The average latency of \HBA~ is 241.79 ms and the standard deviation is 83.06 ms.

\begin{figure}[ht]
\label{fig:hbalatency}
\centering
\includegraphics[width=6cm]{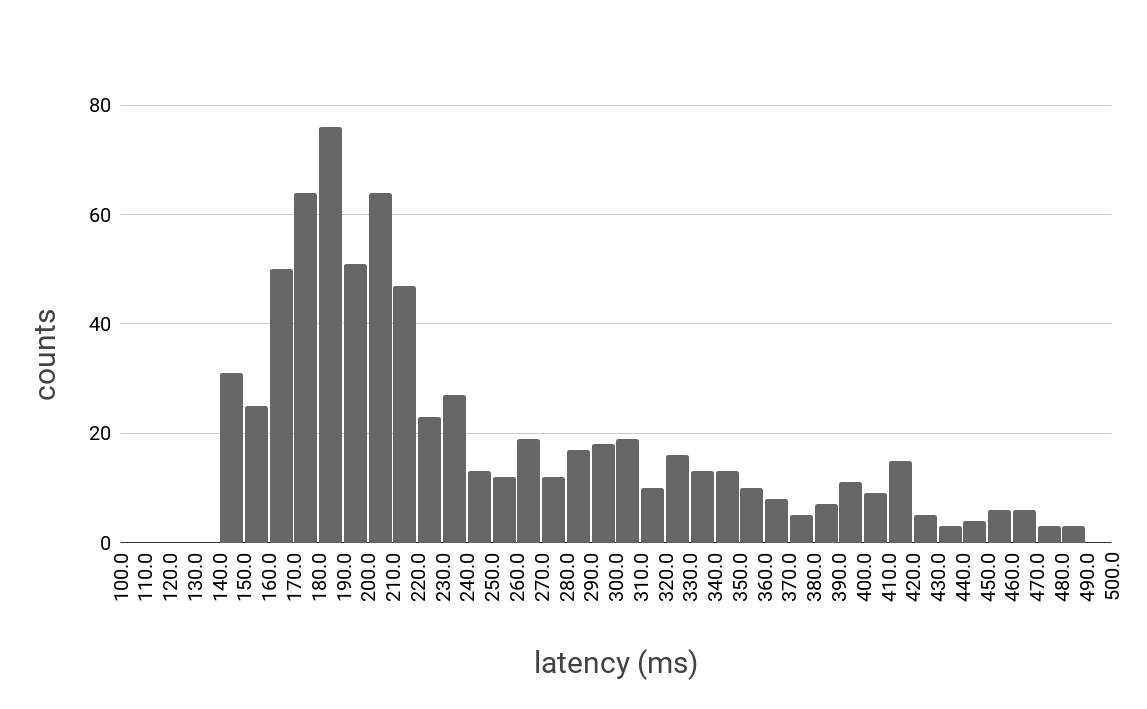}
\caption{The histogram of latency of \HBA.}
\end{figure}

\section{Simulation}
\label{section:simulation}
To demonstrate the performance, responsiveness, and partition-resilience of \HBA, we implement three other Byzantine agreements and compare the simulation results.
The three protocols are PBFT \cite{pbft}, the synchronous BA proposed by Abraham et al. (ADD+19) \cite{AbrahamDDNR18} (the version against static adversary), and Algorand agreement \cite{Chen2018ALGORANDAS}.

Let $n$ be the number of nodes, $f$ be the number of maximum faulty nodes,
and $\lambda$ be the predefined maximum network delay for the protocols.
We implement a network module that each node connected to. 
The actual network delay is parametrized by $\mathcal{G}(\mu, \sigma)$ where the delay is sampled from a Gaussian distribution with mean $\mu$ and standard deviation $\sigma$.


The number of messages sent and the latency during a Byzantine agreement process is recorded from the first message sent to the last node decides its value. Note that we do not have any faulty node in this experiment. We run the experiment on a MacBook Pro with 2.6GHz 6-core Intel Core i7, but the latency is calculated by a simulation clock instead of a wall clock or CPU time, so the result should be able to be reproduced on any machine specification. Means and standard deviations from each result are calculated from 100 times of simulation.



We conduct two experiments to show the behaviors of different protocols under different network conditions.

\paragraph{Responsiveness}
In the first experiment, all the network delays are sampled from $\mathcal{G}(250 ms, 50ms)$.
We execute the four protocols under different $\lambda$ (400ms, 1000ms, 2000ms) and the result is shown in Figure \ref{fig:responsiveness}.

From Figure \ref{fig:responsiveness}, we can see that the confirmation time of BAs with responsiveness such as \HBA~and PBFT only depend on the actual network latency.
Thus, the confirmation time does not change when $\lambda$ varies.
On the other hand, the confirmation time of synchronous BA without responsiveness such as ADD+19 and Algorand agreement increases as $\lambda$ increases.
The ratio between confirmation time and $\lambda$ is the number of total rounds.
It costs around 6.2 rounds and 2.2 rounds for ADD+19 and Algorand agreement, respectively.

\paragraph{Partition-Resilience}
In the second experiment, the network operates in two modes: the normal mode and the partition mode.
In the normal mode, all the nodes are connected with the delay sampled from $\mathcal{G}(250 ms, 50ms)$.
In the partition mode, the network is divided into three distinct sets of size $\lfloor \frac{n}{3} \rfloor$ or $\lfloor \frac{n}{3} \rfloor +1$.
Within the set, the delay is sampled from $\mathcal{G}(250ms, 50ms)$.
For the messages between two sets, the delays are sampled from $\mathcal{G}(4000ms, 1000ms)$.
All the protocols are executed with $\lambda = 1000ms$.
Thus, when the network is in the partition mode, the delay between different sets exceeds $\lambda$.
The protocols are executed in the partition mode for 60 seconds.
Then, the network becomes the normal mode.
The result is shown in Figure \ref{fig:partition1}.

Notice that the partition is ``benign'' in this model.
Except that the delays are sampled from $\mathcal{G}(4000ms, 1000ms)$, there is no adversary that re-schedules or delay the messages to break the protocols maliciously.
The benign partition captures the case that the Internet cables breaks so that the alternative route is saturated.

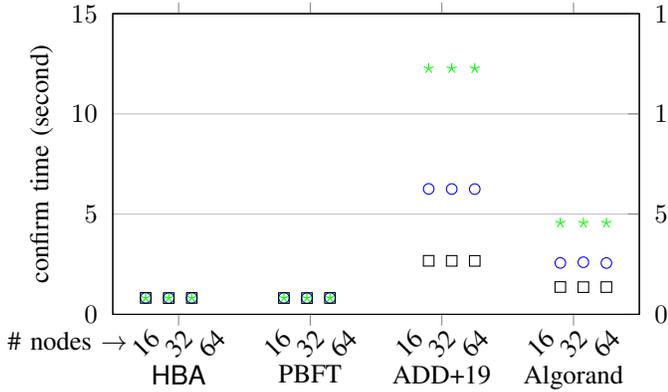
\begin{figure}[ht]
    \centering
    \pgfplotsset{
    MyAxis/.style={
            ybar,
            scale only axis,
            width=7cm,
            height=4cm,
            ymin=0,
            xmin=-0.5,
            xmax=7.5, 
            area legend,
            bar width=6pt
      }
    }
    \begin{tikzpicture}
    \begin{axis}[
        ylabel={confirm time (second)},
        ymax=15,
        ymajorgrids=true,
        xtick={0,0.5,1, 2,2.5,3, 4,4.5,5, 6,6.5,7},
        xtick pos=left,
        ytick pos=left,
        xticklabels = {
           16,32,64,
           16,32,64,
           16,32,64,
           16,32,64
        },
        yticklabel style={xshift=-0.5ex},
        xticklabel style={yshift=0pt,rotate=45}, 
        tickwidth=0pt, 
        extra x ticks={-0.5,2,4.5,7,9.5,12},
        extra x tick labels={},
        minor x tick style = {opacity=0},
        MyAxis
    ]
    \end{axis}

    \begin{axis}[
        name=ax2,
        ymax=15,
        MyAxis,
        ytick pos=right,
        yticklabel style={xshift=0.5ex},
        xtick={0.5,2.5,4.5,6.5},
        xticklabels={\HBA, PBFT, ADD+19, Algorand},
        xticklabel style={yshift=-4mm},
    ]
     \addplot [color=blue, only marks, mark=o,] 
 table[x =x, y =y]{exp2.txt};
      \addplot [color=green, only marks, mark=star,] 
 table[x =x, y =yy]{exp2.txt};
      \addplot [color=black, only marks, mark=square,] 
 table[x =x, y =yyy]{exp2.txt};
    
    \end{axis}
\node at (-0.6, -0.35) {\# nodes $\rightarrow$};
    \end{tikzpicture}
    \caption{Confirmation time of each Byzantine Agreement: The network delay is according to $\mathcal{G}(250 ms, 50ms)$ and the network bound $\lambda = 400$ (resp. 1000 and 2000) ms for the black square (resp. blue circle and green star).}
    \label{fig:responsiveness}
    \end{figure}

\begin{figure}[ht]
    \centering
    \pgfplotsset{
    MyAxis/.style={
            ybar,
            scale only axis,
            width=6cm,
            height=4cm,
            ymin=0,
            xmin=-0.5,
            xmax=7.5, 
            area legend,
            bar width=6pt
      }
    }
    \begin{tikzpicture}
    \begin{axis}[
        ylabel={\# messages},
        ymax=180000,
        ymajorgrids=true,
        xtick={0,0.5,1, 2,2.5,3, 4,4.5,5, 6,6.5,7},
        xtick pos=left,
        ytick pos=left,
        xticklabels = {
           16,32,64,
           16,32,64,
           16,32,64,
           16,32,64
        },
        yticklabel style={xshift=-0.5ex},
        xticklabel style={yshift=0pt,rotate=45}, 
        tickwidth=0pt, 
        extra x ticks={-0.5,2,4.5,7,9.5,12},
        extra x tick labels={},
        minor x tick style = {opacity=0},
        MyAxis
    ]
    \addplot[draw=black,fill=yellow!20] table [x=xa,y=G] {exp.txt};
    \addplot[draw=black,fill=red!20] table [x=xaa,y=H] {exp.txt};
    \addplot[draw=black,fill=green!20] table [x=xaaa,y=I] {exp.txt};
    \addplot[black,sharp plot,dashed]
coordinates {(-0.5,126000) (12,126000)}
;
    \end{axis}

    \begin{axis}[
        name=ax2,
        ymax=80,
        scale only axis,
        width=6cm,
        height=4cm,
        ylabel=confirm time  (second),
        ytick pos=right,
        yticklabel style={xshift=1.5ex},
        xtick={0.5,2.5,4.5,6.5},
        xticklabels={\HBA, PBFT, ADD+19, Algorand},
        xticklabel style={yshift=-5mm},
    ]
    \addplot [color=blue, only marks, mark=o,] plot [error bars/.cd, y dir = both, y explicit]
coordinates {(0,21.664)+-(0,2.411) (0.35,21.978)+-(0,1.911) (0.7,22.743)+-(0,1.187)};
    \addplot [color=purple, only marks, mark=o,] plot [error bars/.cd, y dir = both, y explicit]
coordinates {(2.1,35.484)+-(0,3.664) (2.45,34.88)+-(0,0.769) (2.8,35.782)+-(0,0.731)};
    \addplot [color=black, only marks, mark=o,] plot [error bars/.cd, y dir = both, y explicit]
coordinates {(4.3,70.27)+-(0,0.009) (4.65,70.27)+-(0,0.005) (5,70.267)+-(0,0.004)};
    \addplot [color=green, only marks, mark=o,] plot [error bars/.cd, y dir = both, y explicit]
coordinates {(6.3,64.911)+-(0,1.977) (6.65,66.213)+-(0,0.572) (7,63.836)+-(0,1.842)};
    \end{axis}
\node at (-0.6, -0.35) {\# nodes $\rightarrow$};
\node at (1.3, 3) {network recover};


    \end{tikzpicture}
    \caption{Bandwidth cost and confirmation time of each Byzantine Agreement in partitioned network, where bar chart (left-side) shows the number of message and scatter diagram (right-side) shows the confirm time}
    \label{fig:partition1}
    \end{figure}
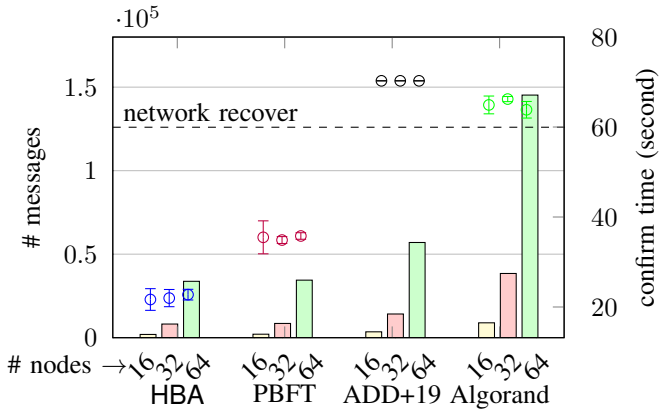

In this experiment, the agreement holds for all the protocols.
From Figure \ref{fig:partition1}, we can see that all four protocols terminate successfully.
In particular, \HBA ~and PBFT terminates before the network is recovered (at 60 seconds).

Concretely speaking, when running \HBA, the honest node $q$ updates $\lockvalue$ when it receives $2\tmax + 1$ pre-commit messages at the current iteration.
In other words, as long as the lock condition is triggered before the forward condition, the honest node $q$ will update $\lockvalue$ and broadcast the commit message of $\lockvalue$ in the next iteration.
Then, honest nodes terminate when they receive $2\tmax + 1$ commit messages.

In order to prevent honest nodes from termination by delaying messages, the adversary needs to trigger the forward condition before the lock condition.
However, such condition rarely happens in practice if the network is not manipulated maliciously.

As for PBFT, the timeout scales up when the view change happens, so once the timeout exceeds the delay, the protocol terminates.
For ADD+19, the protocol is design for the synchronous network, and the partition-resilience is not claimed in their paper, but the protocol terminates after the partition is resolved
\footnote{When the messages are delayed maliciously, the agreement of ADD+19 may be broken. However, such a network condition is beyond their assumption.}.
As Algorand claimed, the protocol terminates immediately after the network is recovered.



\paragraph{Bandwidth Usage}
Finally, we give a short remark to the bandwidth usages.
The numbers of messages are highly related to bandwidth usages.
From Figure \ref{fig:partition1}, the numbers of messages are similar for \HBA ~and PBFT under different participating nodes.
The numbers of messages of ADD+19 and Algorand are more than 69\% and 330\% larger than \HBA~for any setting, respectively.

\section{Conclusion}
\label{section:conclusion}
In this paper, we figure out what counts a suitable BA for blockchains and give the concrete constructions that achieve the properties.
We discuss three desired properties from the aspects of incentive model, security and performance.
The first property is fair validity, and we prove two impossibilities: any BA cannot achieve weakly fair validity in the asynchronous network, and any responsive BA cannot achieve strongly fair validity.
The second property is partition-resilience because the real-world internet is sometimes unstable or attacked by adversaries.
The third property is responsiveness because the latency is usually limited by the time bounds of the synchronous BAs.

We also give two constructions, \RBA ~and \HBA, to demonstrate these properties.
The first protocol, \RBA, achieves strongly fair validity and partition-resilience.
Based on \RBA, the second protocol, \HBA, achieves weakly fair validity, partition-resilience, and responsiveness.
Moreover, comparing to PBFT, \HBA ~enjoys a better resistance to DDoS and better latency in the network partition.
With these properties, \HBA~strikes a balance between fairness, security, and performance.

\bibliographystyle{ieeetr}
\bibliography{BA}

\begin{appendices}

\section{Definition of VRF}
\label{appvrf}
The definition is paraphrased from \cite{vrf}.
\begin{defi}[verifiable random function]
Let $(\keygen, \prove, \veri)$ be a 3-tuple polynomial-time algorithm, where
  \begin{enumerate}
      \item \keygen ~takes as input a security parameter $\kappa$ and outputs a pair of key $(pk,sk)$.
      \item \prove ~takes as input a seed $x$ and a secret key $sk$; it outputs a value $F_{sk}(x)$ and a proof $\pi_{sk}(x)$.
      \item \veri ~takes as input $(pk,x,y,\pi)$; it verifies whether $y = F_{sk}(x)$ by using the proof $\pi$ and key $pk$.
  \end{enumerate}
Let $a:\mathbb{N} \rightarrow \mathbb{N}\cup \{*\}$ and $a:\mathbb{N} \rightarrow \mathbb{N}$ be any functions such that $a(\kappa)$ and $b(\kappa)$ are computable in time $poly(\kappa)$.
We say $(\keygen, \prove, \veri)$ is a verifiable random function with input length $a(\kappa)$ and output length $b(\kappa)$ if the following properties hold:
  \begin{enumerate}
      \item \textbf{Correctness.} If $(y,\pi) = \prove(sk,x)$, then \[
      \Pr[\veri(pk,x,y,\pi) = yes] \geq 1 - \negl(\kappa).
      \]
      \item \textbf{Uniqueness.} For every $(pk,x,y_1,y_2,\pi_1,\pi_2)$ such that $y_1\neq y_2$, the following holds for either $i=1$ or $i=2$: \[
      \Pr[\veri(pk,x,y_i,\pi_i) = yes] \leq \negl(\kappa).
      \]
      \item \textbf{Pseudorandomness.} 
      (Sketched) Any probabilistic polynomial time adversary cannot distinguish the output of a VRF from a uniform random variable.
  \end{enumerate}
\end{defi}

Intuitively, pseudorandomness requires that the output of a VRF should be indistinguishable from a string sampled from a uniform distribution.

\section{Proof in \RBA} \label{app:rba}

\subsection{Agreement}
\begin{lemma*}
Assume $t \leq \tmax$. Suppose a node $p$ receives $2\tmax +1$ commit messages of $v_p$ and another node $q$ receives $2\tmax +1$ commit messages of $v_q$. If both these $2\tmax +1$ commit messages all come from the iteration $r$, then $v_p = v_q$.
\end{lemma*}

\begin{proof}[Proof of lemma \ref{lemma:agreement1}]
We prove this lemma by contradiction. Suppose $v_p \neq v_q$. Because as many as $\tmax$ Byzantine nodes exist, there exists at least one honest node that both commits on $v_p$ and $v_q$ by the pigeonhole principle. 
However, honest nodes can only commit on one value at one iteration, which leads to a contradiction.
\end{proof}

\begin{theorem*}[Agreement]
Assume $t \leq \tmax$ and the adversary is adaptive. Regardless of partition, if an honest node $p$ decides on some value $v_p$ and another honest node $q$ decides on some value $v_q$, then $v_p = v_q$. That is, the honest nodes will never decide on different values.
\end{theorem*}
\begin{proof}[Proof of theorem \ref{thm:baagreement}]
Because $p$ decides on $v_p$ and $q$ decides on $v_q$, $p$ and $q$ must see $2\tmax +1$ commit messages of $v_p$ and $2\tmax +1$ commit messages of $v_q$, respectively. Suppose both these $2\tmax +1$ commit messages come from the same iteration $r$. By Lemma \ref{lemma:agreement1}, we have $v_p = v_q$.

Suppose the $2\tmax +1$ commit messages that $p$ receives come from the iteration $r_p$ and the $2\tmax +1$ commit messages that $q$ receives come from the iteration $r_q$. Without loss of generality, we assume $r_p < r_q$. Because there are up to $\tmax$ Byzantine nodes, there must be at least $\tmax+1$ honest nodes commit on $v_p$ so that $p$ can receive $2\tmax +1$ commit messages of $v_p$. 
For all iterations $r > r_p$, these $\tmax+1$ honest nodes will always pre-commit on $v_p$ until they see $2\tmax +1$ pre-commit messages of $v' \neq v_p$. 
However, only $2\tmax$ nodes remain, so these $\tmax+1$ honest nodes will never pre-commit any $v' \neq v_p$ for all $r > r_p$. Thus, for all $r > r_p$, if some value $v$ has $2\tmax +1$ pre-commit messages, then $v = v_p$. 

Because $q$ receives $2\tmax +1$ commit messages of $v_q$, there must exist at least $\tmax+1$ honest nodes that commit on $v_q$ at the iteration $r_q$. These $\tmax+1$ honest nodes commit on $v_q$ only if they have seen $2\tmax +1$ pre-commit messages of $v_q$ at iteration $r_q$. Therefore, $v_q = v_p$.
\end{proof}

\subsection{Termination}
\begin{proposition*}[Termination without partition in adaptive adversary]
Assume $t \leq \tmax$ and the adversary is adaptive. If all the honest nodes start at the $r$-th iteration within time $\lambda$ and no partition exists, all the honest nodes will decide on some values in $t+1$ iterations. 
\end{proposition*}
\begin{proof}[Proof of proposition \ref{proposition:terminationWOpartition}]
In this proof, we divide all the possibilities into three cases. First, we suppose there is an honest node has decided on some value. Second, we suppose that no honest node has decided, but there exists an honest node has seen $2\tmax +1$ pre-commit messages of the same value. The third case includes all the else possibilities.

\noindent\textit{Case 1: Some honest node has decided.}
If an honest node $p$ has decided on the value $v_p$, $p$ must have seen $2\tmax +1$ commit messages of $v_p$. Because $p$ propagates these $2\tmax +1$ commit messages, all the honest nodes will hold this information after time $\lambda$ and decide on $v_p$ in one iteration.

\noindent\textit{Case 2: Some honest node has seen $2\tmax +1$ pre-commit messages on the same value.}
Suppose no node has decided but there exists an honest node $p$ that has seen $2\tmax +1$ pre-commit messages of a value $v_p$. Because $p$ propagates these $2\tmax +1$ pre-commit messages, all the honest nodes will hold this information after time $\lambda$. 
With these $2\tmax +1$ pre-commit messages, all the honest nodes update their internal variables $\lockvalue = v_p$ according to the condition $1$.
Consequently, all the honest nodes will pre-commit on $v_p$ at the next iteration and thus commit on $v_p$ as well. 
At the end of the next iteration, they will all decide on $v_p$.

\noindent\textit{Case 3: Else possibilities.}
Because no honest node has ever seen $2\tmax +1$ pre-commit messages, $\lockvalue = \bot$ for all honest node $q$. Thus, they will identify their leader by their local view.
Because all honest nodes start at the $r$-th iteration within time $\lambda$, they can receive all the initial values from other honest nodes before identifying the leaders. 
Thus, there exist some honest nodes that pre-commit different values relative to each other only if a Byzantine node proposes different initial values to different nodes\footnote{Note that not proposing any initial value is considered to be equivalent to proposing $\bot$.}. However, the honest nodes will propagate the initial value so all honest nodes will have the same set of initial values after time $\lambda$. Thus, to prevent the honest nodes from agreeing on the same leader, Byzantine nodes must propose different initial values to different nodes at every iteration. However, a node can only propose an initial value once, or it will be caught. Thus, the best strategy of Byzantine nodes is that different Byzantine nodes propose their initial values at different iterations so $t$ Byzantine nodes can only interfere during $t$ iterations. Thus, all the honest nodes will decide on some values in $t+1$ iterations with certainty.
\end{proof}

\begin{proposition*}[Expected termination in static adversary]
Assume $t \leq \tmax$ and the adversary is static.
Suppose all honest nodes start at $r$-th iteration within time $\lambda$ and no partition exists. Then, it is expected that all honest nodes will decide on some values in $8$ rounds.
\end{proposition*}
\begin{proof}[Proof of proposition \ref{proposition:terminationExpected}]
From the proof of Proposition \ref{proposition:terminationWOpartition}, we know that if some honest node has decided on value $v$ or has seen $2\tmax +1$ pre-commit messages of a value $v$, then all the honest nodes will decide on $v$ in one iteration.

In a network without partition, the best strategy for the Byzantine nodes has been described in Case 3 in the proof of Proposition \ref{proposition:terminationWOpartition}. However, to interfere with $k$ iterations successfully, the Byzantine nodes must win the leadership in the following $k$ iterations. The probability of such an event is \[
\prod_{i=0}^{k-1} \left(\frac{t-i}{n-i}\right)(\frac{n-t}{n}) \leq \left(\frac{t}{n}\right)^k(\frac{n-t}{n}).
\]
Thus, in expectation, the number of rounds can be computed by \[
\sum_{i=0}^t \left(\frac{t}{n}\right)^i(\frac{n-t}{n}) \cdot (6+4i) \leq 6+\frac{4 \cdot \frac{t}{n}}{(1-\frac{t}{n})}.
\]
Because $n \geq 3\tmax+1 \geq 3t+1$, the expected number of rounds is 8. 
\end{proof}

\begin{proposition*}[Fast recovery from partition in adaptive adversary]
Assume $t \leq \tmax$ and the adversary is adaptive.
If the partition is resolved, all the honest nodes will decide on some values in $t+2$ iteration. 
If the adversary is static, it is to be expected that all honest nodes will decide on some values in $12$ rounds.
\end{proposition*}
\begin{proof}[Proof of proposition \ref{proposition:partitionResolved}]
If there exists a node $p$ that has decided on a value $v_p$, $p$ must have seen $2\tmax +1$ commit messages of $v_p$. All the honest nodes will receive these $2\tmax +1$ commit messages of $v_p$ within time $\lambda$ after the partition is resolved and decide on $v_p$.

Suppose no node has decided and $p$ is the node working on the latest iteration $r_p$. To enter the iteration $r_p$, $p$ must achieve the forward condition at iteration $r_p - 1$. Because the partition is resolved, all honest nodes will also achieve the forward condition within time $\lambda$ after the partition is resolved and also enter the iteration $r_p$. Later on, if some node $q$ achieves the forward condition and enters the iteration $r_p + 1$, other honest nodes will also achieve the forward condition within time $\lambda$. Thus, all honest nodes start at the iteration $r_p + 1$ with time difference $< \lambda$ and Proposition \ref{proposition:terminationWOpartition} guarantees that they will decide on some values within the following $t+1$ iterations. 

Because each iteration costs $4\lambda$, similarly, if the adversary is static, it is to be expected that all honest nodes will decide on some values in $8 + 4$ rounds according to Proposition \ref{proposition:terminationExpected}.
\end{proof}

\subsection{Strongly Fair Validity}
We first show that the probability is exactly lower-bounded by the uniform distribution in the ideal world.
Then, we show that \RBA ~works the same as the ideal world except the negligible probability.

We define the VRF oracle, consisting of two algorithm: $\Oprove$ and $\Overi$.
$\Oprove$ is defined as:
\begin{enumerate}
    \item Take as input a seed $x$ and a secret key $sk$.
    \item Return $\prove(x,sk)$.
\end{enumerate}
$\Overi$ is defined as:
\begin{enumerate}
    \item Take as input a public key $pk$, a seed $x$, a value $y$ and a proof $\pi$.
    \item Return $\veri(pk,x,y,\pi)$.
\end{enumerate}
We also define the ideal functionality of VRF, consisting of two algorithm: $\Iprove$ and $\Iveri$.
$\Iprove$ is defined as:
\begin{enumerate}
    \item Takes as input a seed $x$ and a secret key $sk$.
    \item Check whether $Q(x,sk)$ is defined. If not, choose $y\leftarrow \{0,1\}^\ell$ and $\pi\leftarrow \{0,1\}^\ell$ uniformly at random. Then, set $Q(x,sk) = (y,\pi)$.
    If $Q(x,sk)$ is defined, $\Iprove(x,sk)$ return $Q(x,sk)$.
\end{enumerate}
$\Iveri$ is defined as:
\begin{enumerate}
    \item Takes as input a public key $pk$, a seed $x$, a value $y$ and a proof $\pi$.
    \item Check whether $Q(x,sk)$ is defined. If not, return false; otherwise, return true.
\end{enumerate}

In the ideal world, all the nodes does not compute and verify the value of VRF locally.
Instead, they query the oracle $\Iprove$ and $\Iveri$.
All the else operations are the same as \RBA.

\begin{lemma}[fairness in the ideal world]
\label{lemma:idealworld}
Suppose the network is synchronous.
Then, in the ideal world, for all adversaries and for all $q \in \mathcal{H}$, conditioned on all the honest nodes have decided on some values, it holds that
\begin{align}
    \Pr[q\text{'s value is the decided by some honest node}] \geq \frac{1}{n}.
\end{align}
\end{lemma}

\begin{proof}[Proof of lemma \ref{lemma:idealworld}]
Because the VRF value $y_i$ are chosen uniformly at random for all nodes $i$ in the ideal world, the probability that the node $q$ wins the minimum value among $\min_{i\in \{1,\cdots,n\}}y_i$ (the leadership) is exact $\frac{1}{n}$.

Once the node $q$ wins the leadership, all the honest nodes will broadcast the pre-commit messages on $v_q$ at $2\lambda$ and broadcast the commit messages on $v_q$ at $4\lambda$ because the network is synchronous.
In this case, $q$'s value will be decided by all hones nodes.
\end{proof}

\begin{theorem*}[strongly fair validity]
Suppose the network is synchronous and $F$ is a secure VRF.
Then, \RBA ~achieves strongly fair validity under the assumption of static adversary.
\end{theorem*}
\begin{proof}[Proof of theorem \ref{thm:rbafair}]
We prove it by the hybrid argument.
Let $\Hyb_1$ be the protocol the same as the ideal world except that the node $q_1$ queries $\Oprove$ and $\Overi$ instead of $\Iprove$ and $\Iveri$, respectively.
Then, for all $i\in\{2,\cdots,n\}$, let $\Hyb_i$ be the protocol the same as $\Hyb_{i-1}$ except that the node $q_i$ queries $\Oprove$ and $\Overi$ instead of $\Iprove$ and $\Iveri$, respectively.

Because $F$ is a secure VRF, the behavior of $(\Oprove,\Overi)$ is indistinguishable from $(\Iprove,\Iveri)$.
Thus, the ideal world is indistinguishable from $\Hyb_1$.
Similarly, for all $i\in\{2,\cdots,n\}$, $\Hyb_{i-1}$ is indistinguishable from $\Hyb_i$.
Because $n$ is bounded by $poly(\kappa)$, the ideal world is indistinguishable from $\Hyb_n$.

Then, the honest nodes in \RBA ~always compute VRF correctly.
So, there is no different for the honest nodes that whether the VRF is computed locally or is queried by $(\Oprove,\Overi)$.
Therefore, $\Hyb_n$ is indistinguishable from \RBA.

Combining the arguments above, we have that the ideal world is indistinguishable from \RBA.
That is, there exists a negligible function $\eta$ such that \RBA ~works the same as the ideal world except the negligible probability $\eta(\kappa)$.
Let $X_q$ be the event that $q$'s value is the decided by some honest node conditioned on \RBA ~works the same as the ideal world for the node $q$.
Let $\overline{X_q}$ be the event that $q$'s value is the decided by some honest node conditioned on \RBA ~does not work the same as the ideal world for the node $q$.
Combine the result with Lemma \ref{lemma:idealworld}, we have that in \RBA, for all $q \in \mathcal{H}$,
\begin{align}
  &\Pr[q\text{'s value is the decided by some honest node}]\\
= &\Pr[X_q] \cdot (1-\eta(\kappa)) + \Pr[\overline{X_q}] \cdot \eta(\kappa) \\
\geq &\frac{1}{n} \cdot (1-\eta(\kappa)).
\end{align}
Thus, \RBA ~achieves strongly fair validity.
\end{proof}

\section{Proof in \HBA} \label{app:hba}
\subsection{Agreement}
\begin{theorem*}[Agreement of \HBA]
Assume $t \leq \tmax$ and the adversary is adaptive. Regardless of partition, if an honest node $p$ decides on some value $v_p$ and another honest node $q$ decides on some value $v_q$, then $v_p = v_q$. That is, the honest nodes will never decide on different values.
\end{theorem*}
\begin{proof}[Proof of theorem \ref{theorem:agreementhba}]
The proof is almost the same as the proof of Theorem \ref{thm:baagreement}. 
For completeness, we state the formal proof here.
We call the commit message with the timestamp $0$ (sent in Step 2) comes from the iteration $0$.
Hence, for each iteration, an honest node can only commit on one value.

Because $p$ decides on $v_p$ and $q$ decides on $v_q$, $p$ and $q$ must see $2\tmax +1$ commit messages of $v_p$ and $2\tmax +1$ commit messages of $v_q$, respectively. Suppose both these $2\tmax +1$ commit messages come from the same iteration $r$. According to the proof of Lemma \ref{lemma:agreement1}, we have $v_p = v_q$.

Suppose the $2\tmax +1$ commit messages that $p$ receives come from the iteration $r_p$ and the $2\tmax +1$ commit messages that $q$ receives come from the iteration $r_q$. Without loss of generality, we assume $r_p < r_q$. Because there are up to $\tmax$ Byzantine nodes, there must be at least $\tmax+1$ honest nodes commit on $v_p$ so that $p$ can receive $2\tmax +1$ commit messages of $v_p$. 
For all iterations $r > r_p$, these $\tmax+1$ honest nodes will always pre-commit on $v_p$ until they see $2\tmax +1$ pre-commit messages of $v' \neq v_p$. 
However, only $2\tmax$ nodes remain, so these $\tmax+1$ honest nodes will never pre-commit any $v' \neq v_p$ for all $r > r_p$. Thus, for all $r > r_p$, if some value $v$ has $2\tmax +1$ pre-commit messages, then $v = v_p$. 

Because $q$ receives $2\tmax +1$ commit messages of $v_q$, there must exist at least $\tmax+1$ honest nodes that commit on $v_q$ at the iteration $r_q$. These $\tmax+1$ honest nodes commit on $v_q$ only if they have seen $2\tmax +1$ pre-commit messages of $v_q$ at iteration $r_q$. Therefore, $v_q = v_p$.

\end{proof}

\subsection{Termination}

\begin{proposition*}[Termination without partition in static adversary]
Assume $t \leq \tmax$ and the adversary is static. 
If all the honest nodes start \HBA~within time $\lambda$ and no partition exists, all the honest nodes will decide on some values in $4\lambda$, $6.33\lambda$ and $t+1$ iterations in the best case, the average case and the worst case, respectively.
\end{proposition*}
\begin{proof}[Proof of proposition \ref{proposition:staticTermination}]
We categorize into two cases:\\
\noindent\textit{Case 1: The pioneer is an honest nodes.}
Because the leader is honest, it will broadcast the message at the beginning of \HBA.
All the honest nodes receive the leader's value $v_\ell$ and reply in $2\lambda$.
Then, all the honest nodes receive $2\tmax + 1$ pre-commit message in $3\lambda$.
Meanwhile, they broadcast the commit messages on $v_\ell$.
Thus, all the honest nodes receive $2\tmax + 1$ commit messages on $v_\ell$ and terminate within $4\lambda$, which is the best case.

\noindent\textit{Case 2: The pioneer is a Byzantine node.}
By theorem \ref{theorem:agreementhba}, all honest nodes decide either in Step 2 or in Step 4-6 of some iteration.
In the former case, all honest will reach the termination condition within $4 \lambda$.
In the latter case, all honest nodes terminate in $t+1$ iterations in the worst case according to Proposition \ref{proposition:terminationWOpartition}.

In expectation, they terminate in $8$ rounds according to Proposition \ref{proposition:terminationExpected}.
Since the probability of the leader in fast phase is honest node is $2/3$, the expected time of termination is
\[
\frac{2}{3}\cdot 4\lambda + \frac{1}{3} \cdot (3\lambda+8\lambda) = 6.33 \lambda.
\]
\end{proof}

\begin{proposition*}[Termination without partition in adaptive adversary]
Assume $t \leq \tmax$ and the adversary is adaptive. 
If all the honest nodes start \HBA~within time $\lambda$ and no partition exists, all the honest nodes will decide on some values in $t+1$ iterations.
\end{proposition*}
\begin{proof}[Proof of proposition \ref{proposition:terminationWOpartitionfastBA}]
Because the leader in Step 2 is pre-determined, the adaptive adversary can always compromise the leader.
This is the worst case in Case 2. in Proposition \ref{proposition:staticTermination} and we have that all the honest nodes will terminate in $t+1$ iterations by Proposition \ref{proposition:terminationWOpartition}.
\end{proof}

\begin{proposition*}[Fast recovery from a partition in adaptive adversary]
Assume $t \leq \tmax$ and the adversary is adaptive. If the partition is resolved, all the honest nodes will decide on some values in $t+2$ iterations. 
If the adversary is static, it is to be expected that all honest nodes will decide on some values in $12$ rounds.
\end{proposition*}
\begin{proof}[Proof of proposition \ref{proposition:hbafastrecoveryfrompartition}]
Suppose some honest nodes have decided in the fast mode. Then, after the partition is resolved, they would broadcast the proof, and all honest nodes will terminate and agree on the value proposed in the fast mode in $\lambda$.
If no honest node has decided in the fast mode, then all the honest nodes proceed to the normal mode. In this case, the termination property is exactly the same as \RBA~and we have proved it in Proposition \ref{proposition:partitionResolved}.
\end{proof}

\subsection{Responsiveness}
\begin{proposition*}
Assume the actual network delay is $\delta$, and all the nodes start \HBA~within $\tau$ time difference. 
If there is no partition and the pioneer is honest, all the honest nodes will decide on some values in $\tau + 3\delta$.
\end{proposition*}
\begin{proof}[Proof of proposition \ref{proposition:hbaresponsiveness}]
Suppose all the honest nodes start \HBA~simultaneously. The honest pioneer broadcasts its value at $\clock = 0$. All the honest nodes will receive pioneer's value and reply the pre-commit messages in $\delta$. All the honest nodes will receive $2\tmax + 1$ pre-commit messages before $\clock = 2\delta$. 
Because the pioneer is honest, these $2\tmax + 1$ messages all pre-commit on the same value.
Thus, all the honest nodes broadcast the commit messages and will decide in $3\delta$.

If the pioneer broadcasts its value at $\clock = \tau$ for some node $q$ due to the time difference, $q$ will decide at $\clock = \tau + 3\delta$.
\end{proof}

\subsection{Wearkly Fair Validity}
\begin{theorem*}[weakly fair validity]
Suppose the network is synchronous.
Then, \HBA ~achieves weakly fair validity under the assumption of static adversary.
\end{theorem*}
\begin{proof}[Proof of theorem \ref{thm:hbafairness}]
When the node $q$ is elected as the pioneer, because the network is synchronous, all the honest nodes will receive $q$'s fast message and broadcast the pre-commit messages on $v_q$ before $2\lambda$ (we allow honest nodes start the protocol within $\lambda$ time drift).
Then, all the honest nodes will receive $2\tmax + 1$ pre-commit messages on $v_q$ before $3\lambda$, so they all set their $\lockvalue$ on $v_q$.
In this case, they will all terminates on $v_q$.
Thus, as long as the network is synchronous, honest nodes will always terminates on honest pioneer's value.

Because the pioneer is elected by the permutation of nodes' public keys, all the nodes will be the pioneer once if \HBA ~is executed $n$ times.
Except that the adversary can forge the signature (only with negligible probability), all the honest nodes can propose a value that be decided by all honest nodes at least $\lfloor \frac{M}{n}\rfloor$ times after \HBA is executed $M$ times.
Thus, \HBA ~achieves weakly fair validity.
\end{proof}

\section{Complexity Analysis and Discussion}
\label{section:complexity}
 We analyze the communication complexity for \RBA~and \HBA~in Section \ref{subsec:RBA_complexity} and Section \ref{subsec:HBA_complexity}, respectively.

\subsection{Communication Complexity of Robust Byzantine Agreement}
\label{subsec:RBA_complexity}
We now analyze the communication complexity of a single node for a single round in \RBA. 
Because an honest node will help to propagate the messages, all honest nodes will gossip $\mathcal{O}(n)$ messages in a single round. 
Thus, the communication complexity for all nodes is $\mathcal{O}(n^2)$ in a single round.

As discussed in Section \ref{subsection:baTermination}, if no partition exists or the system recovers from a partition, \RBA~terminates in $t+1$ iterations in the worst case and is expected to terminate in $8$ rounds. 
We assume $n \geq 3t+1$, so the protocol terminates in $\mathcal{O}(n)$ iterations in the worst case and is expected to terminate in $\mathcal{O}(1)$ iterations. Therefore, the total communication complexity of the protocol is $\mathcal{O}(n^3)$ in the worst case and $\mathcal{O}(n^2)$ in the expected case.

\subsection{Communication Complexity of Hybrid Byzantine Agreement}
\label{subsec:HBA_complexity}
Since the core of \HBA~is actually \RBA, except that every node first enters a fast voting procedure.
The average-case and worst-case communication complexity remain the same.
For the best-case, the upper bound of the communication 
complexity is $n + n^2 + n^2 = 2n^2 + n$.

\subsection{Communication-Efficient Recovery}
Nodes can either actively request data or passively receive data
while nodes suspect a partition happened.
As mentioned in Section \ref{section:model}, if a node $q$ recovers from a partition, it should receive all the previous messages which should be delivered.
In this subsection, we argue that other nodes are not necessary to send all the received messages but only the messages that certify the newest status.

Precisely, a node $q$ sends the commit message of $v$ with $2\tmax+1$ signatures from different nodes to certify that $q$ decides on $v$,
the pre-commit message with $2\tmax+1$ signatures from different nodes to certify that $q$ locks on the certain round and $v$, or
pre-commit message with $2\tmax+1$ signatures from different nodes to certify the latest iteration at which $q$ working.

Furthermore, the communication cost can be reduced by using a threshold signature to compact the $2\tmax+1$ signatures into constant size.

\end{appendices}

\end{document}